\documentclass[lettersize,journal]{IEEEtran}
\usepackage{amsmath,amsfonts}
\usepackage{amsthm}
\usepackage{algorithmic}
\usepackage{algorithm}
\usepackage{array}
\usepackage[caption=false,font=normalsize,labelfont=sf,textfont=sf]{subfig}
\usepackage{textcomp}
\usepackage{stfloats}
\usepackage{url}
\usepackage{verbatim}
\usepackage{graphicx}
\usepackage{hyperref}
\usepackage{xcolor} 
\hypersetup{colorlinks, linkcolor={black}, citecolor={black}, urlcolor={black}}
\usepackage{cite}
\newtheorem{theorem}{Theorem}
\newtheorem{lemma}[theorem]{Lemma}

\newtheorem{problem}{Problem}

\newcommand{\abs}[1]{\lvert#1\rvert}
\newcommand\numberthis{\addtocounter{equation}{1}\tag{\theequation}}

\makeatletter
\renewcommand{\make@df@tag@@@}[2][]{%
	\gdef\df@tag{%
		\tagform@{#2\rlap{\hphantom)#1}}%
		\toks@\@xp{\p@equation{#2}}%
		\edef\@currentlabel{\the\toks@}%
	}%
}
\makeatother

\begin{document}

\title{D\&A: Resource Optimisation in Personalised PageRank Computations Using Multi-Core Machines}

\author{Kai Siong Yow$^{\ast}$ and Chunbo Li
\thanks{Kai Siong Yow is with the School of Computer Science and Engineering, Nanyang Technological University, Singapore and the Department of Mathematics and Statistics, Faculty of Science, Universiti Putra Malaysia, 43400 UPM Serdang, Selangor, Malaysia. E-mail: kaisiong.yow@ntu.edu.sg, ksyow@upm.edu.my}
\thanks{Chunbo Li is with the School of Computer Science and Engineering, Nanyang Technological University, Singapore. E-mail: chunbo001@e.ntu.edu.sg}
\thanks{$^{\ast}$Correspondence:~\texttt{ksyow@upm.edu.my}}}



\maketitle

\begin{abstract}
	Resource optimisation is commonly used in workload management, ensuring efficient and timely task completion utilising available resources. It serves to minimise costs, prompting the development of numerous algorithms tailored to this end. The majority of these techniques focus on scheduling and executing workloads effectively within the provided resource constraints. In this paper, we tackle this problem using another approach. We propose a novel framework D\&A to determine the number of cores required in completing a workload under time constraint. We first preprocess a small portion of queries to derive the number of required slots, allowing for the allocation of the remaining workloads into each slot. We introduce a scaling factor in handling the time fluctuation issue caused by random functions. We further establish a lower bound of the number of cores required under this scenario, serving as a baseline for comparison purposes. We examine the framework by computing personalised PageRank values involving intensive computations. Our experimental results show that D\&A surpasses the baseline, achieving reductions in the required number of cores ranging from $ 38.89\% $ to $ 73.68\% $ across benchmark datasets comprising millions of vertices and edges.
\end{abstract}


\section{Introduction}
\IEEEPARstart{R}{esource} optimisation~\cite{Chen_JCC2021,Zhou_MPE2022} stands as a pivotal technique in efficiently managing and allocating workloads based on available resources. Its application spans across diverse domains, ranging from project management to cloud computing. Specifically, in computer science, it seeks to enhance the execution of programs or minimise resource usage such as memory, power, or other resources, during computations.

A variety of resource optimisation algorithms have been designed for different purposes, including the round-robin scheduling (RR)\cite{Rasmussen_EJOR2008}, ant colony optimisation algorithm (ACO)~\cite{Dorigo_IEEE1997,Zhou_MPE2022}, particle swarm optimisation algorithm (PSO)~\cite{Kennedy_ICNN1995}, genetic algorithm~\cite{Holland_SA1992} and bacterial foraging optimisation algorithm\cite{Nasir_ESA2015}. The majority of current algorithms prioritise discovering efficient methods for scheduling workloads according to available resources, aiming to achieve task completion within a reduced time frame. For further details on optimisation techniques, readers are referred to~\cite{Pham_2012}.

We first provide a concise overview of several existing resource optimisation techniques commonly employed in cloud environments. The heuristic algorithm RR first assigns a fixed time slot and processes each task in a circular order. If a task is not completed after the assigned time slot, it will be interrupted and the subsequent task in the queue will be processed. The process terminates when all tasks are completed. This technique is however not suitable for large-scale platforms.
For ACO, it is used to solve computational problems based on probabilistic techniques where an agent (artificial ant) moves through all possible solutions in order to locate an optimal solution by modelling problems using graphs. It is designed based on the concept that pheromones produced by real ants guide other ants to traverse between food and their colony, and ultimately generate the optimal path. In the cloud computing environment, ants act as agents that travel between machines and jobs, and assign requested jobs to cloud resources. Zuo et al.~\cite{Zuo_IEEE2016} developed a multi-objective scheduling method based on ACO to minimise the completion time and cost.
The metaheuristic algorithm PSO optimises a problem via a population (swarm) of candidate solutions known as particles, where particles move around the search space according to mathematical formula that is developed based on the position and velocity of particles. It requires just the objective function and makes very few or no assumptions. Fatima et al.~\cite{Fatima_ELEC2018} proposed an enhanced PSO that combines PSO and a levy flight algorithm to minimise the number of physical servers under a cloud environment.

Acknowledging the significance of resource optimisation techniques, especially within cloud environments, we approach this problem from a fresh angle in this work. Rather than concentrating on scheduling workloads effectively through heuristic methods, our proposal centers on a framework aimed at minimising the number of cores necessary to fulfill a given task within predefined time constraints.
We focus on the computations of \emph{personalised PageRank} (PPR)~\cite{Page_1999,Wang_KDD2020,Ivan_bio2011}, an extension of PageRank (PR)~\cite{Brin_1998,Page_1999,Berkhin_2005} algorithm, which is introduced to cope with the prediction of personal preferences. \IEEEpubidadjcol
Given a graph $ G $ and a pair of vertices $ s $ and $ t $, the PPR value $ \pi(s,t) $ is defined as the probability of a random walk that begins from the \emph{source vertex} $ s $ and ends at the \emph{target vertex} $ t $. This value indicates the importance of $ t $ with regard to $ s $. The concept has been applied in various fields~\cite{Park_IEEE2019} including recommendation systems, social network analysis and information retrieval.

Computing PPR values is crucial, yet costly, particularly for graphs with high orders $ n $. This process entails extracting eigenvalues from an $ n \times n $ matrix. In addition, $ O(n^{{2}}) $ space is required in order to store all PPR values given that distinct pairs of source and target vertices derive different PPR values.
Due to the nature of the problem, PPR computations using various heuristic and approximation methods~\cite{Fogaras_IM2005,Lofgren_KDD2014,Lofgren_WSDM2016,Wang_PVLDB2016,Wang_KDD2017} remain challenging, even to solve its relax version by fixing the source and target vertices. Hence, we handle this problem by computing PPR values within a parallel computing environment~\cite{Hou_PVLDB2021}, with the objective of minimising core usage and ensuring the processing of all queries within the specified time constraint, rather than reducing computation time.


We aim to address the following optimisation problem.
\begin{problem}\label{prob:minimising_cores}
	Given a multi-core machine with $ \mathcal{C} $ cores, design a computational model so that a huge number $ \mathcal{X} $ of personalised PageRank queries can be processed within a given time $ \mathcal{T} $, by minimising $ \mathcal{C} $.
\end{problem}

We design an approach namely \textsc{Divide and Allocate in parallel computing} (abbreviated as \textsc{D\&A}) in addressing Problem~\ref{prob:minimising_cores}. This approach will later be extended by fixing the number of available resources in dealing with real-world scenarios. Resource optimisation using multi-core machines can potentially be applied in various domains including server and data centre management, database query optimisation, and scientific and engineering applications. To the best of our knowledge, we are pioneering an investigation into the problem from this perspective, where the primary objective is to minimise and determine an appropriate utilisation of cores. To facilitate comparison, we additionally establish a theoretical bound to ascertain the lower bound of the number of cores for the problem under identical constraints. Our framework exhibits flexibility in the sense that it can be employed by any algorithm feasible within parallel environments.

We validate our framework by computing PPR values using FORA~\cite{Wang_KDD2017} (an effective index-based solution that merges forward push and Monte Carlo random walk techniques) on a multi-core machine, based on different benchmark datasets with up to $ 4.8 \times 10^{6} $ vertices and $ 6.8 \times 10^{7} $ edges. The algorithm is chosen due to its proven superiority in managing large-scale networks and its widespread recognition for efficiency. Furthermore, the utilisation of random functions within the framework adds complexity to the problem. Our experimental results demonstrate the effectiveness of the framework, where the number of cores is reduced by up to $ 73.68\% $ compared to the outcomes obtained through the theoretical approach.

The subsequent sections of this article are outlined as follows: We present a sample size formula in Section~\ref{sec:sample_size} to determine the number of samples required during the preprocessing phase. Subsequently, in Sections~\ref{sec:D&A} and \ref{sec:theoretical_cound}, we introduce our proposed model and establish a theoretical bound to serve as a baseline for comparative purposes. Section~\ref{sec:experimental_setting} delineates the experimental settings, followed by an exposition of our results and analyses in Section~\ref{sec:results}. Finally, we draw conclusions and suggest potential avenues for future research in Section~\ref{sec:conclusion}.

\section{Sample Size Estimation}\label{sec:sample_size}

The optimisation of the number of cores is heavily dependent on the average running time $ \bar{t} $ in processing each query. When dealing with a large volume of queries, a subset of sample queries can be used during the preprocessing step in determining $ \bar{t} $. Several methods can be employed for this task, such as published tables~\cite{Krejcie_1970}, historical data in similar studies or employing formulae~\cite{Cochran_1977}.

Rather than relying on a random selection of samples, our approach involves initially utilising a sample size formula to ascertain the requisite number of samples for estimating $ \bar{t} $ during the preprocessing stage. This ensures statistical reliability and precision in our estimations. Additionally, we aim to achieve a delicate equilibrium between acquiring a suitably large sample size to diminish sampling error and minimising unnecessary data collection to optimise resource utilisation.

For a large number of queries,  we use the following formula~\cite{Cochran_1977} to determine the sample size $ s $ required during the preprocessing stage:
\begin{equation}\label{eqn:sample_size}
	s = \frac{Z^{2} \cdot p \cdot (1-p)}{e^{2}} 
\end{equation}
where $ Z, p, $ and $ e $ denote the standard score (also known as $ z $-score), population proportion that is normally distributed and acceptable sampling error, respectively. The formula gives the lower bound of the required sample size with some margin of error in the estimated proportion.
In Equation~\ref{eqn:sample_size}, the value of $ z $-score is associated with the chosen confidence interval where the most commonly chosen confidence intervals are 90\%, 95\% and 99\%. The value of $ p $ is in between 0 to 1, and if no any information is available to approximate $ p $, the most conservative way is to set $ p=0.50 $ to generate the largest sample size. The value of $ e $ is usually expressed as a percentage, which implies the desired degree of precision. If the sample size is too small (respectively, big), the confidence interval can be increased (respectively, reduced) or the acceptable sampling error can be reduced (respectively, increased).

\paragraph{Example}

Given a 99\% confidence interval, 0.50 population proportion and an acceptable sampling error of 5\%, we have
\begin{equation}\label{eqn:sample_size_example}
	s = \frac{2.576^{2} \cdot 0.50 \cdot (1-0.50)}{0.05^{2}} = 663.58 \approx 664.
\end{equation}

\section{The Proposed Model}

\subsection{D\&A Framework}\label{sec:D&A}

We first design a general approach namely \textsc{Divide and Allocate in parallel computing} (abbreviated as \textsc{D\&A}) as shown in Algorithm~\ref{alg:minimum_cores}, to determine the number of cores required in order to complete $ \mathcal{X} $ independent queries within a given duration $ \mathcal{T} $. Here, we assume that $ \mathcal{X} $ is large and there is no any constraint on available resources. See Fig.~\ref{fig:D&A} for an illustration of the framework.
\begin{algorithm}[t]
	\caption{\textsc{D\&A} ($ \mathcal{X}, \mathcal{T} $)}
	\begin{algorithmic}[1]
				
		\STATE Using Equation~\ref{eqn:sample_size}, determine the sample size $ s $ by setting $ x\% $ confidence interval and $ e $ acceptable sampling error.
		\STATE Preprocess $ s $ queries in parallel using $ s $ cores.
		\STATE Let $ t_{i} $ be the time to process the $ i^{th} $ query, and set $ t_{max} = \max\{t_{i} \vert i=1,\ldots,s\} $.
		\STATE Let $ \ell = \lfloor \frac{\mathcal{T}-t_{max}}{t_{max}} \rfloor $ be the slots allocated to the remaining duration.
		\STATE Assign $ k = \lceil \frac{\mathcal{X}-s}{\ell} \rceil $ queries to each of the $ \ell $ slots. For each slot, process all $ k $ queries in parallel using at most $ k $ cores.
		\STATE For $ j \in \{1,\ldots,k\} $, let $ T_{j} $ be the total processing time for the $ j^{th} $ core to complete all the allocated queries in each slot. 
		\STATE Set $ T_{max} = \max\{T_{j} \vert j=1,\ldots,k\} $.
		\STATE \textbf{if} {$ t_{max} + T_{max} \le \mathcal{T} $}
		\STATE \hspace{0.5cm} \textbf{return} $ k $.
		\STATE \textbf{else}
		\STATE \hspace{0.5cm} Go back to Step 2.
	\end{algorithmic}
	\label{alg:minimum_cores}
\end{algorithm}

\begin{figure}[t]
	\centering
	\includegraphics[scale=0.21]{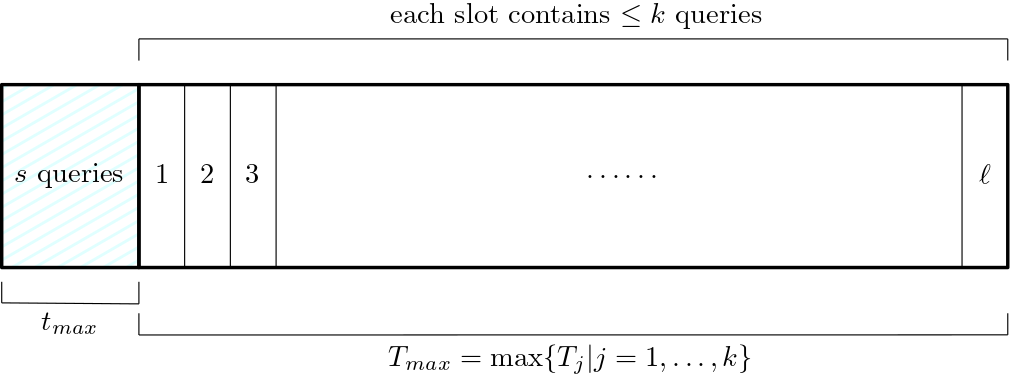}
	\caption{An illustration of D\&A}
	\label{fig:D&A}
\end{figure}

To determine a suitable sample size $ s $ based on Equation~\ref{eqn:sample_size}, we set the confidence interval and acceptable sampling error at $ x\% $ and $ e $, respectively (Line 1).
We also set $ p=0.50 $ to generate the largest sample size. Let $ t_{i} $ represent the processing time to complete the $ i^{th} $ query and $ t_{max} = \max\{t_{i} \vert i=1,\ldots,s\} $. To minimise the overall processing time, all $ s $ queries will be preprocessed in parallel using $ s $ cores, and $ t_{max} $ will then be determined (Lines 2--3). We determine the number of slots $ \ell = \lfloor \frac{\mathcal{T}-t_{max}}{t_{max}} \rfloor $ based on the remaining duration $ \mathcal{T}-t_{max} $. The floor function ensures that the total processing time is at most $ \mathcal{T} $ when all the remaining queries are processed in parallel in their respective slots, assuming that the processing time in each slot is at most $ t_{max} $ (Line 4).

For each slot, we evenly assign to it $ k = \lceil \frac{\mathcal{X}-s}{\ell} \rceil $ queries, ensuring that all remaining queries are allocated, as denoted by the ceiling function. We also assume that $ s \leq k $; otherwise, no further action is necessary, as $ s $ cores are utilised during the preprocessing stage, hence requiring $ s $ cores. As a consequence of the ceiling function, some slots may contain less than $ k $ queries. Subsequently, all $ k $ queries in each slot are processed in parallel, using at most $ k $ cores (Line 5). For $ j \in \{1,\ldots,k\} $, let $ T_{j} $ represent the total processing time for each $ j^{th} $ core to complete all the allocated queries in each slot (Line 6). If $ t_{max} + \max\{T_{j}\} \le \mathcal{T} $ and considering $ s \leq k $, we require a minimum of $ k $ cores to process all $ \mathcal{X} $ queries within time $ \mathcal{T} $ (Lines 7--9).
Alternatively, if the total running time surpasses $ \mathcal{T} $, indicating that the task cannot be accomplished within the allotted time frame, we revert to Step 2 and repeat the entire process (Lines 10--11). This situation may arise due to significant fluctuations in processing times when queries are processed within their respective slots.

We now derive a lower bound of $ \mathcal{C} $ based on Algorithm~\ref{alg:minimum_cores}.

\begin{lemma}\label{lem:bound_c}
	For a multi-core machine without any constraint on the number of cores, suppose $ \mathcal{X} $ is the number of queries, $ s $ is the number of sample queries and $ t_{i} $ is the running time to complete the $ i^{th} $ sample query. If $ t_{max} = \max\{t_{i} \vert i=1,\ldots,s\} $, then the minimum number of cores required to process all $ \mathcal{X} $ queries within time $ \mathcal{T} $ is at least $ \frac{\mathcal{X} \cdot t_{max}}{\mathcal{T}} $.
\end{lemma}

\begin{proof}
	Let $ \mathcal{X}, \mathcal{T} $ and $ t_{max} $ be as stated. Given a multi-core machine that has no constraint on the number of cores, all queries can be processed in parallel where each query is processed by one core at a time. Suppose $ s $ queries are used in the preprocessing step in obtaining $ t_{max} $. By using $ t_{max} $, we now determine the number $ \ell $ of slots and hence the number $ k $ of queries that can be allocated into each slot.
	
	Based on Algorithm~\ref{alg:minimum_cores}, if the total running time in completing all $ \mathcal{X} $ queries is at most $ \mathcal{T} $, then we need at least $ k $ cores to process all the queries. This implies that
	\begin{equation}\label{eqn:lower_bound_c}
		k \ge (\mathcal{X}-s) \cdot \left( \frac{t_{max}}{\mathcal{T}-t_{max}} \right).
	\end{equation}
	Note that $ s $ is the number of sample queries used in the preprocessing step. We replace $ s $ by $ k $ in Equation~\ref{eqn:lower_bound_c}, so that the number of queries during each parallel computation (including the preprocessing step) is as balanced as possible. By rearranging and solving Equation~\ref{eqn:lower_bound_c}, we have
	\begin{align*}
		k \cdot \left( 1 + \frac{t_{max}}{\mathcal{T}-t_{max}} \right) & \ge \mathcal{X} \cdot \left( \frac{t_{max}}{\mathcal{T}-t_{max}} \right)\\
		k & \ge \mathcal{X} \cdot \left( \frac{t_{max}}{\mathcal{T}-t_{max}} \right)\left( \frac{\mathcal{T}-t_{max}}{\mathcal{T}} \right)\\
		k & \ge \frac{\mathcal{X} \cdot t_{max}}{\mathcal{T}}. \numberthis \label{eqn:lower_cound_c_final}
	\end{align*}
\end{proof}

Given the impracticality of having unlimited resources in real-world environments, and considering that the required number of cores depends on both the processing time of each query and a specified duration (by Equation~\ref{eqn:lower_cound_c_final}), we adapt the principle utilised in Algorithm~\ref{alg:minimum_cores} and extend the algorithm by restricting the number of available cores.

We use Lemma~\ref{lem:bound_c} in estimating the minimum number of cores required. By Equation~\ref{eqn:lower_cound_c_final}, it is evident that $ \mathcal{T} $ plays a crucial role in determining $ k $, to compute $ \mathcal{X} $ queries where the sample queries computed in parallel have the maximum running time $ t_{max} $. This suggests that a specified duration dictates the minimum number of cores necessary for computation. It is worth noting that with a fixed number of available cores, there exists a maximum capacity in processing a given task under a time constraint.

In the real-world environment, we could first estimate the minimum number $ k $ of cores based on a given duration. If $ k $ is higher than the available cores, then the process will be terminated. Otherwise, we prolong the duration to ensure that a feasible solution can always be obtained (since the number of cores is restricted).

By modifying Algorithm~\ref{alg:minimum_cores}, we propose \textsc{D\&A\_Real} (see Algorithm~\ref{alg:minimum_cores_practical}) that can be adapted into the real-world environment. Apart from the number $ \mathcal{X} $ of queries and a given duration $ \mathcal{T} $, we input the number of available cores $ C_{max} $. 
\begin{algorithm}[t]
	\caption{\textsc{D\&A\_Real} ($ \mathcal{X}, \mathcal{T}, C_{max} $)}
	\begin{algorithmic}[1]
		
		\STATE Preprocess $ s $ sample queries using $ c \ll s $ cores.
		\STATE For $ i=\{1,\ldots,s\} $, let $ t_{i} $ be the time to process the $ i^{th} $ query. Set $ t_{max} = \max\{t_{i}\} $, $ t_{pre} = \sum t_{i} $ and $ t_{avg} = \frac{c \cdot t_{pre}}{s} $.
		\STATE Using Equation~\ref{eqn:lower_cound_c_final}, compute the lower bound $ C $ of the number of cores required to process $ \mathcal{X} $ queries within time $ \mathcal{T} $.
		\STATE \textbf{if} {$ C_{max} < \lceil C \rceil $}
		\STATE \hspace{0.5cm} \textbf{raise} Error. 
		\STATE \textbf{else}
		\STATE \hspace{0.5cm} Let $ \ell = \lfloor \frac{d \cdot \mathcal{T}-t_{pre}}{t_{avg}} \rfloor $ be the slots allocated to the remaining duration, where $ d \le 1 $ is a scaling factor.		
		\STATE \hspace{0.5cm} Assign $ k = \lceil \frac{\mathcal{X}-s}{\ell} \rceil $ queries to each of the $ \ell $ slots. For each slot, process all $ k $ queries in parallel using at most $ k $ cores.		
		\STATE \hspace{0.5cm} For $ j \in \{1,\ldots,k\} $, let $ T_{j} $ be the total processing time for the $ j^{th} $ core to complete all the allocated queries in each slot.		
		\STATE \hspace{0.5cm} Set $ T_{max} = \max\{T_{j} \vert j=1,\ldots,k\} $.
		\STATE \hspace{0.5cm} \textbf{if} {$ t_{pre} + T_{max} \le T $}
		\STATE \hspace{1.0cm} \textbf{return} $ k $. 
		\STATE \hspace{0.5cm} \textbf{else}
		\STATE \hspace{1.0cm} \textbf{raise} Error. 
	\end{algorithmic}
	\label{alg:minimum_cores_practical}
\end{algorithm}

To optimise resources in the real-world environment where the number of cores is restricted, the proposed concept closely resembles Algorithm~\ref{alg:minimum_cores}, albeit with a few modifications. When preprocessing $ s $ sample queries, instead of using $ s $ cores, we employ $ c \ll s $ cores (Line 1), since employing a relatively large $ s $ may improve efficiency but conflicts with our objective of minimising core usage. Hence, we set $ c=1 $ in our experiments. Subsequently, we use Equation~\ref{eqn:lower_cound_c_final} to determine the lower bound $ C $ of the number of cores, ensuring that $ \lceil C \rceil $ does not exceed $ C_{max} $ in the worst-case scenario (Lines 3--5). We then proceed with the divide and allocate operations outlined in Algorithm~\ref{alg:minimum_cores} to determine the minimum number of cores required, with appropriate modifications (Lines 6--14).

\subsection{Theoretical Bound}\label{sec:theoretical_cound}

We now determine the lower bound of $ \mathcal{C} $ in Problem~\ref{prob:minimising_cores} by using Hoeffding's inequality~\cite{Hoeffding_JASA1963}, which will be used as the baseline for comparison purposes.

\begin{lemma}\label{lem:hoeffdings_inequality}
	In Problem~\ref{prob:minimising_cores}, we have
	\begin{equation}\label{eqn:hoeffdings_inequality}
		\mathcal{C} \ge \frac{\mathcal{X}}{\mathcal{T}} \cdot \left( \bar{t_{k}} + \sqrt{\frac{\hat{t^{2}} \ln (2/p_{f})}{2k}} \right)
	\end{equation}
	where $ \bar{t_{k}} $ denotes the average processing time $ t $ of $ k $ sample queries, $ \hat{t} $ denotes the upper bound of $ t $, and $ p_{f} $ is a failure probability.
\end{lemma}

\begin{proof}
	Let $ p_{f} $ be as stated. We now estimate the lower bound of the number $ \mathcal{C} $ of cores required to process all $ \mathcal{X} $ queries within a given time $ \mathcal{T} $. We assume that $ \mathcal{X} $ is sufficiently large and all query times $ t $ have a same distribution. Let $ E[t] $ be the expected value of $ t $, we aim to
	\begin{equation*}
		minimise\ \mathcal{C}
	\end{equation*}
	subject to
	\begin{equation}\label{eqn:constraint}
		P\left( \frac{\mathcal{X} \cdot E[t]}{\mathcal{C}} \le \mathcal{T} \right) \ge 1-p_{f}.
	\end{equation}
	
	To obtain the estimation on $ \mathcal{C} $, we first process $ k $ PPR queries. Let $ t_{1},t_{2},\ldots,t_{k} $ be the processing times for the $ k $ queries, and $ \bar{t_{k}} = \frac{1}{k} \sum_{i=0}^{k} t_{i} $. Suppose the lower and upper bounds of the query time are $ 0 $ and $ \hat{t} $, respectively. We apply Hoeffding's inequality to estimate the lower bound of $ \mathcal{C} $. For all $ \lambda>0 $, we first have
	\begin{equation*}
		P\left( \abs{\bar{t_{k}} - E[t]} \ge \lambda \right) \le 2 \exp \left( -\frac{2k^{2}\lambda^{2}}{k\hat{t}^{2}} \right).
	\end{equation*}
	The constraint in Equation~\ref{eqn:constraint} is satisfied when $ \lambda \le \frac{\mathcal{T}\mathcal{C}}{\mathcal{X}} - \bar{t_{k}} $. Suppose $ \lambda = \frac{\mathcal{T}\mathcal{C}}{\mathcal{X}} - \bar{t_{k}} $. To satisfy Equation~\ref{eqn:constraint}, we have
	\begin{align*}
		2 \exp \left( -\frac{2k(\frac{\mathcal{T}\mathcal{C}}{\mathcal{X}} - \bar{t_{k}})^{2}}{\hat{t}^{2}} \right) & \le p_{f}\\
		\frac{2k(\frac{\mathcal{T}\mathcal{C}}{\mathcal{X}} - \bar{t_{k}})^{2}}{\hat{t}^{2}} & \ge \ln \left( \frac{2}{p_{f}} \right)\\
		\frac{\mathcal{T}\mathcal{C}}{\mathcal{X}} - \bar{t_{k}} & \ge \sqrt{\frac{\hat{t}^{2} \ln(2/p_{f})}{2k}}\\
		\mathcal{C} & \ge \frac{\mathcal{X}}{\mathcal{T}} \cdot \left( \bar{t_{k}} + \sqrt{\frac{\hat{t^{2}} \ln (2/p_{f})}{2k}} \right).
	\end{align*}
	Hence, by satisfying the constraint in Equation~\ref{eqn:constraint}, we obtain the lower bound of $ \mathcal{C} $ for Problem~\ref{prob:minimising_cores}.
\end{proof}

\section{Experiments}

In this section, we evaluate the effectiveness of \textsc{D\&A\_Real} by computing PPR queries using FORA~\cite{Wang_KDD2017}. The experimental results are compared with the bound derived in Lemma~\ref{lem:bound_c}.

\subsection{Experimental settings}\label{sec:experimental_setting}

\paragraph{Setup} 

We conduct our experiments on a server with an Intel(R) Xeon(R) Gold 6326 CPU@2.90GHz processor, 256GB memory and running 64-bit Ubuntu 20.04.4 LTS. There are 64 available cores on the server.

\paragraph{Data and query sets} 

We focus on four benchmark datasets \emph{Web-Stanford}, \emph{DBLP}, \emph{Pokec} and \emph{LiveJournal} in conducting the experiments, as summarised in Table~\ref{table:datasets}. 
\begin{table}[H]
	\caption{Summary of datasets}
	\label{table:datasets}
	\centering
	\begin{tabular}{lrrc}
		\hline
		\multicolumn{1}{c}{Dataset} & \multicolumn{1}{c}{Order ($ n $)} & \multicolumn{1}{c}{Size ($ m $)} & \multicolumn{1}{c}{Type}\\
		\hline
		\emph{Web-Stanford} & 281,903 & 2,312,497 & Directed\\
		\emph{DBLP} & 613,586 & 3,980,318 & Undirected\\
		\emph{Pokec} & 1,632,803 & 30,622,564 & Directed\\
		\emph{LiveJournal} & 4,847,571 & 68,993,773 & Directed\\
		\hline	
	\end{tabular}
\end{table}

The number of sample queries during the preprocessing stage is mainly determined by the order and size of the respective graph, the processing time for each query and the number of cores assigned. For \emph{Web-Stanford}, we use the formula as discussed in Section~\ref{sec:sample_size} to determine the number $ s $ of sample queries, where $ s $ is set conservatively as in Equation~\ref{eqn:sample_size} based on the number of queries used for experimental purposes. For graphs with larger orders and sizes (\emph{DBLP}, \emph{Pokec}, and \emph{LiveJournal}), we however observe that the same strategy is less suitable to determine the number of sample queries, due to a longer processing time per query. Hence, the number of sample queries is fixed at 5\% of the smallest number of queries for large graphs, given that a relatively small $ c $ will be used during the preprocessing step in the real-world environment.
It is noteworthy that the average running time per query during the preprocessing stage remains consistent across these datasets throughout the experiments, under both scenarios.

\paragraph{Parameters} To be more conservative, we set $ c = 1 $ during the preprocessing stage so that the number of required cores in processing a given number of queries within a given duration can be minimised under different scenarios. The duration $ \mathcal{T} $ is set based on the processing time per query derived in\cite{Wang_KDD2017}. Recall that random functions used in FORA may cause fluctuations of processing time, hence we also include a scaling factor $ d \le 1 $ in our framework in dealing with this concern. Instead of taking the average processing time by repeating the experiments, we notice that $ d $ can be used in coping with the fluctuation issue by setting it appropriately.

Since the time fluctuation is directly proportional to the order and size of a graph, the value of $ d $ should behave the other way round. Thus, we set $ d = 1.00 $ for \emph{Web-Stanford}, $ d = 0.85 $ for both \emph{DBLP} and \emph{Pokec}, and $ d = 0.80 $ for \emph{LiveJournal}. We follow~\cite{Wang_KDD2017} for the values of all other parameters in FORA.

\subsection{Results and Analyses}\label{sec:results}

We evaluate the effectiveness of the proposed framework \textsc{D\&A\_Real} based on FORA using a multi-core machine. The outcome of our experiments is shown in Fig.~\ref{fig:result}, in which the number $ \mathcal{X} $ of queries and given durations $ \mathcal{T} $ are set differently for various datasets.
\begin{figure*}[!t]
	\centering
	\subfloat[\emph{Web-Stanford} when $ \mathcal{T}=50 $]{
		\includegraphics[width=1.735in]{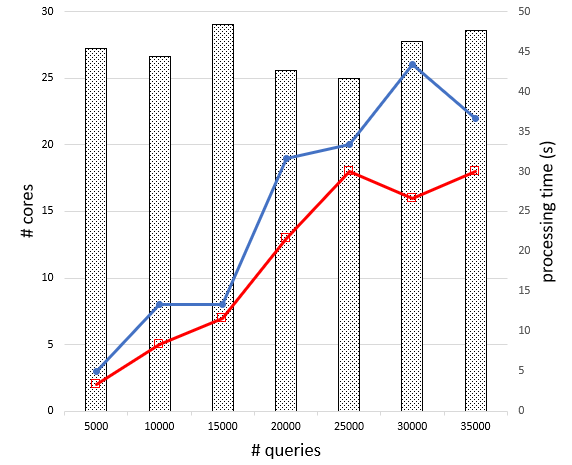}%
	}
	\hfil
	\subfloat[\emph{DBLP}  when $ \mathcal{T}=500 $]{
		\includegraphics[width=1.735in]{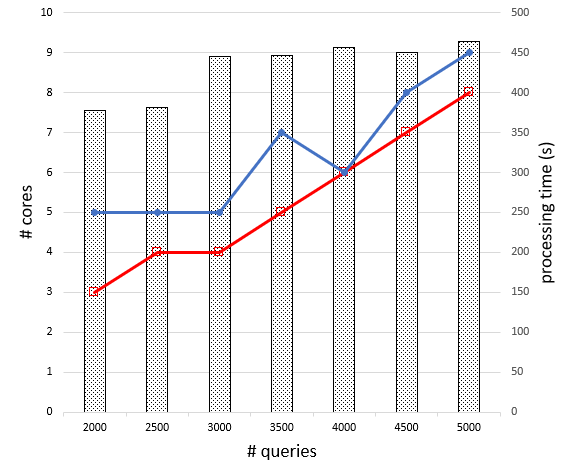}%
	}
	\hfil
	\subfloat[\emph{Pokec}  when $ \mathcal{T}=600 $]{
		\includegraphics[width=1.735in]{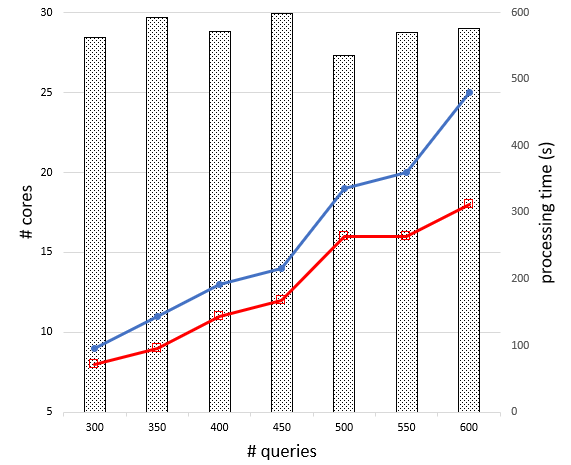}%
	}
	\hfil
	\subfloat[\emph{LiveJournal}  when $ \mathcal{T}=800 $]{
		\includegraphics[width=1.735in]{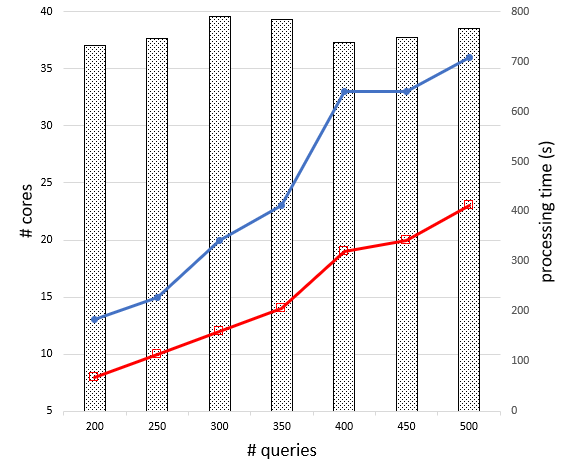}%
	}
	\caption{Results for the minimum number of required cores based on four benchmark datasets, by varying the number $ \mathcal{X} $ of queries. The bar chart indicates the processing time whereas the line graphs indicate the number of required cores for \textsc{D\&A\_Real} (in red) and the theoretical bound in Lemma~\ref{lem:hoeffdings_inequality} (in blue)}
	\label{fig:result}
\end{figure*}

Based on Fig.~\ref{fig:result}, we can see that the number of cores required by \textsc{D\&A\_Real} is always lesser comparing to the lower bound that we derived theoretically in Lemma~\ref{lem:hoeffdings_inequality} (since the bound is also affected by the average and upper bound of processing times), except for one circumstance where they both are equal under the \emph{DBLP} dataset (see Fig.~\ref{fig:result}(b)). This phenomenon primarily arises from the random functions inherent in FORA, leading to a lower upper bound $ \hat{t} $ for the set of running times when handling sample queries compared to others. Consequently, this results in a smaller boundary in Equation~\ref{eqn:hoeffdings_inequality}. Since the processing time of FORA is affected by the random functions used for generating random walks, we believe that if the value of $ \hat{t} $ is analogous in all circumstances, similar results will always be obtained where the number of cores required by \textsc{D\&A\_Real} is lesser than the theoretical bound. We also observe that a smaller or the same number of cores is sometimes required even when the number of queries is higher. This is again possible if $ \hat{t}_{1} > \hat{t}_{2} $ for $ \mathcal{X}_{1} < \mathcal{X}_{2} $, where $ \hat{t}_{i} $ is associated to $ \mathcal{X}_{i} $.

The experimental results show that \textsc{D\&A\_Real} is effective in minimising the number of cores required under parallel computing environments when it is tested using FORA. It reduces the number of cores by up to $ 62.50\%, 66.67\%, 38.89\% $ and $ 73.68\% $ for four benchmark datasets \emph{Web-Stanford}, \emph{DBLP}, \emph{Pokec} and \emph{LiveJournal}, respectively. We expect a more consistent result if \textsc{D\&A\_Real} is tested by frameworks that involve lesser random functions, in which the number of slots can be determined in a more accurate manner.

We then make a comparison by manipulating the scaling factor $ d $ that is used to address the time fluctuation issue in our experiments. For \emph{Web-Stanford}, we observe that if $ d $ is reduced from $ 1.00 $ to $ 0.85 $ with all other variables remain, all queries can be completed in a shorter duration with a higher number of cores under most instances (see Fig.~\ref{fig:result_comparison}).
\begin{figure}[t]
	\centering
	\subfloat[$ d = 1.00 $]{
		\includegraphics[width=1.695in]{web_T50_02}%
	}
	\hfil
	\subfloat[$ d = 0.85 $]{
		\includegraphics[width=1.695in]{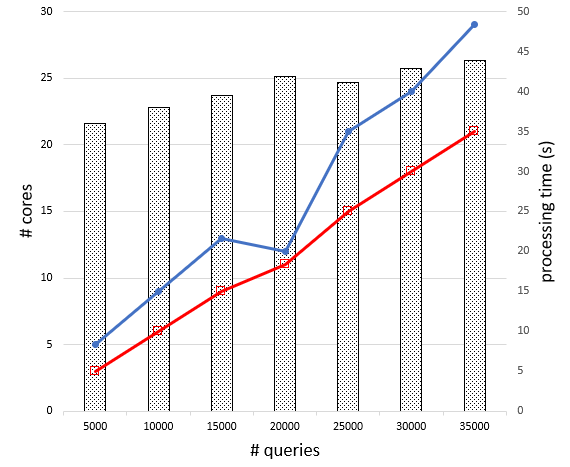}%
	}
	\caption{A comparison using different scaling factor $ d $ for \emph{Web-Stanford}}
	\label{fig:result_comparison}
\end{figure}
This is in accordance with our intention in which a lower $ d $ will lead to a smaller number of slots and hence a higher number of cores. For all other datasets, we found that if $ d = 1.00 $, then all queries may not be processed within the given duration due to the time fluctuations caused by random functions, which further indicates the benefit of the scaling factor. We note that graphs with larger orders and sizes require a smaller $ d $ to process all queries within a given time frame.  It is hence natural to ask if there is a method to determine $ d $ theoretically. We assert that the values of $ d $ are optimal for both \emph{Pokec} and \emph{LiveJournal} datasets given that their processing times for certain queries are close to the maximum (as shown in Fig.~\ref{fig:result}(c) and Fig.~\ref{fig:result}(d)) as per our experimental findings.

\section{Conclusion and Future Direction}\label{sec:conclusion}

We design a framework \textsc{D\&A\_Real} to establish the necessary number of cores for processing a certain number of queries within a fixed duration in parallel computing environments. Our approach involves dividing and allocating queries into specific slots based on the processing time identified during the preprocessing stage.  To mitigate the variability introduced by randomness in the framework, instead of conducting repeated experiments, we introduce a scaling factor. We measure the effectiveness of the framework using FORA, and compare the experimental results to a theoretically derived baseline. Our findings demonstrate that \textsc{D\&A\_Real} outperforms the baseline, as it requires significantly fewer cores even for graphs with millions of vertices and edges.

For future work, we propose expanding our framework by integrating it into diverse algorithms with different functionalities (including online settings) within parallel environments, so that its effectiveness can be validated further, aligning with real-world requirements. Additionally, learning-based approaches~\cite{Yow_arxiv2023v3,Yow_PVLDB2023} may also be employed to enhance the overall performance of this framework.

\section*{Acknowledgments}

The authors would like to thank all the reviewers for their constructive comments and guidance, to express gratitude to Dingheng Mo for sharing his insight in deriving the theoretical bound, to Siqiang Luo for suggesting this direction, and to Ningyi Liao for his helpful comments. The first author is most grateful to Singapore National Academy of Science for appointing him as an SASEA Fellow through a grant (NRF-MP-2022-0001) supported by NRF Singapore, and Universiti Putra Malaysia for granting him Leave of Absence in completing this work in Nanyang Technological University.


\bibliographystyle{IEEEtran}
\bibliography{references}

\end{document}